\documentclass[11pt,a4paper]{amsart}
\usepackage{amsmath,amssymb, amsbsy}
\usepackage{subfigure}
\usepackage{graphpap,latexsym,epsf}
\usepackage{color,psfrag}
\usepackage[dvips]{graphicx}
\usepackage{enumerate}
\usepackage{bbm}
\usepackage{relsize}
\usepackage{hyperref}
\hypersetup{
    colorlinks=true,
    linkcolor=blue,
    filecolor=magenta,
    urlcolor=cyan,
}
\begin{document}
\newcommand{\dyle}{\displaystyle}
\newcommand{\R}{{\mathbb{R}}}
\newcommand{\Hi}{{\mathbb H}}
\newcommand{\Ss}{{\mathbb S}}
\newcommand{\N}{{\mathbb N}}
\newcommand{\Rn}{{\mathbb{R}^n}}
\newcommand{\F}{{\mathcal F}}
\newcommand{\ieq}{\begin{equation}}
\newcommand{\eeq}{\end{equation}}
\newcommand{\ieqa}{\begin{eqnarray}}
\newcommand{\eeqa}{\end{eqnarray}}
\newcommand{\ieqas}{\begin{eqnarray*}}
\newcommand{\eeqas}{\end{eqnarray*}}
\newcommand{\f}{\hat{f}}
\newcommand{\Bo}{\put(260,0){\rule{2mm}{2mm}}\\}
\newcommand{\1}{\mathlarger{\mathlarger{\mathbbm{1}}}}

%- Theorems and similar stuff: ------

\theoremstyle{plain}
\newtheorem{theorem}{Theorem} [section]
\newtheorem{corollary}[theorem]{Corollary}
\newtheorem{lemma}[theorem]{Lemma}
\newtheorem{proposition}[theorem]{Proposition}
\def\neweq#1{\begin{equation}\label{#1}}
\def\endeq{\end{equation}}
\def\eq#1{(\ref{#1})}

%- Definitions -----------------------------

\theoremstyle{definition}
\newtheorem{definition}[theorem]{Definition}
\newtheorem{remark}[theorem]{Remark}
\numberwithin{figure}{section}

\title[It\^o vs Stratonovich in cosmology]{It\^o versus Stratonovich \\ in a stochastic cosmological model}

\author[C. Escudero, C. Manada]{Carlos Escudero, Carlos Manada}
\address{}
\email{}


\date{\today}

\begin{abstract}
In this work we study a stochastic version of the Friedmann acceleration equation. This model has been proposed in the cosmology literature
as a possible explanation of the uncertainty found in the experimental quantification of the Hubble parameter. Its noise has been tacitly
interpreted in the Stratonovich sense. Herein we prove that this interpretation leads to a positive probability of finite time blow-up of
the solution, that is, of the Hubble parameter. In contrast, if we just modify the noise interpretation to that of It\^o, then the solution
globally exists almost surely. Moreover, the expected asymptotic behavior is found under this interpretation too.
\end{abstract}
\maketitle

\section{Introduction}

The Hubble parameter is a measure of the rate at which the universe expands. As such, it is a fundamental quantity in cosmology.
Despite its importance, or perhaps as a consequence of it, there is currently an uncertainty in its experimental quantification.
In reference \cite{john}, this uncertainty is theoretically approached by means of a stochastic generalization of the Friedmann
acceleration equation; that is, the Hubble parameter is considered to be a stochastic process rather than a deterministic function of time.
For the sake of completeness we offer a derivation of the model studied in \cite{john};
to this end we follow \cite{chamizo}, \cite{foster}, \cite{john}, and \cite{misner}. Let us
start considering the Einstein field equations
\begin{equation}
\label{einsteineq}
R_{\mu\nu}=8\pi G\left(T_{\mu\nu}-\frac{1}{2}Tg_{\mu\nu}\right),
\end{equation}
with the Friedmann-Lema\^itre-Robertson-Walker metric
\begin{equation}
\label{robertsonwalker}
ds^2=-dt^2+C^2(t)\left(\frac{dr^2}{1-kr^2}+r^2d\theta^2+r^2\sin^2\theta d\varphi^2\right),
\end{equation}
and the energy-momentum stress tensor
\begin{equation}
\label{stresstensor}
T_{\mu\nu}=(p+g)\delta_{\mu}^{0}\delta_{\nu}^{0}+pg_{\mu\nu}.
\end{equation}
Herein $G$ denotes the Newton gravitational constant, $g_{\mu\nu}$ the metric tensor, $C(t)$ the scale factor,
$k$ is the Gaussian curvature, and $p$ the pressure.
If $R_{\mu\nu}$ is the Ricci tensor for metric \eqref{robertsonwalker}, we have $R_{00}=-3\frac{C''}{C}$, and from \eqref{stresstensor}, $T=T_{\lambda}^{\lambda}=T_{\lambda\nu}g^{\nu\lambda}=3p-\rho$, and $T_{00}=\rho+p+pg_{00}=\rho$, where $\rho$ is the mass density; then from \eqref{einsteineq} we find
\begin{equation}
\label{cos1p0}
-3\frac{C''}{C}=4\pi G(\rho+3p).
\end{equation}
If we introduce the Hubble parameter $H:=\frac{C'}{C}$, this equation becomes the Friedmann acceleration equation
\begin{equation*}
\dot{H}=-H^2-\frac{4\pi G}{3}(\rho+3p).
\end{equation*}
On the other hand for $\mu=\nu=1$, we have $R_{11}=\frac{CC''+2(C')^2+2k}{1-kr^2}$ and $T_{11}=pg_{11}=\frac{C^2p}{1-kr^2}$, and
from \eqref{einsteineq} we find
\begin{equation}
\label{cos1p1}
CC''+2(C')^2+2k=4\pi G(\rho-p)C^2.
\end{equation}
Now multiply \eqref{cos1p0} by $C^2/3$ and sum it to \eqref{cos1p1} to obtain the Friedmann equation
\begin{equation}
\label{cos1p2}
(C')^2+k=\frac{8\pi G}{3}\rho \, C^2.
\end{equation}
From \eqref{cos1p1} and \eqref{cos1p2} we deduce that $p$ and $\rho$ only depend on $t$. If we moreover restrict ourselves to flat models (i.e. $k=0$),
and assume the proportionality $p=\xi \rho$, with $\xi$ a parameter yet to be specified, we get
\begin{equation*}
\dot{H}=-\frac{3}{2}H^2 \left(1+\xi \right).
\end{equation*}
If we were to consider a deterministic model, and furthermore, following \cite{john}, a model in which the condition $3p+\rho=0$ is fulfilled,
then we had to conclude $\xi=-\frac{1}{3}$. Then $C''=0$ by \eqref{cos1p0}, and therefore, if $C(0)>0$ and $C'(0)>0$, then $H(0)>0$ and
$$
H(t) = \frac{H_0}{1 + H_0 t},
$$
where $H_0:=H(0)$ and we have
$$H(t) t \rightarrow 1 \qquad \text{when} \quad t\rightarrow\infty.$$
Some of the properties of this model, such as uniqueness and global existence of the solution, along with its
asymptotic behavior, could be expected to hold in a stochastic counterpart of it. It is one of the goals of this work
to specify some conditions that guarantee so.

Now, again following \cite{john}, if we assume that the condition $3p+\rho=0$ has only to be fulfilled on average, we could write
$\sigma \xi'=\xi+\frac{1}{3}$, where $\xi'$ could be any zero mean stochastic process. For simplicity, and as in \cite{john}, consider
it is a standard Gaussian white noise to arrive at the model
\begin{equation*}
\dot{H}=-\frac{3}{2}H^2 \left(\frac{2}{3}+ \sigma \xi' \right),
\end{equation*}
where the noise amplitude $\sigma>0$ has been introduced for dimensional consistency.
This is a stochastic version of the Friedmann acceleration equation that can be slightly simplified by
introducing the rescaling $t \to \sigma^2 t$ and $H \to H/\sigma^2$ to find
\begin{equation}
\label{cos3}
\dot{H}=-\frac{3}{2}H^2 \left(\frac{2}{3}+ \xi' \right),
\end{equation}
where the rescaled time and Hubble parameter are obviously dimensionless (and their dimensional counterparts are trivially recovered
by undoing the rescaling).
Of course, this model is not precise, and it could only be made so by means of the introduction of a notion of stochastic integration.
This type of matter has been studied in much detail in physics \cite{lefever}. In \cite{john}, the authors tacitly (not explicitly) assume
the Stratonovich interpretation of noise, known to be the only one able to preserve the classical chain rule from ordinary differential calculus.
However, as we will show herein, this choice is somehow problematic. Precisely, it opens the possibility of finite time blow-ups for the Hubble parameter,
phenomena of difficult (or at least not completely straightforward) physical interpretation,
while this problematic disappears when the It\^o interpretation of noise is chosen.
Before starting with this program let us emphasize that equation \eqref{cos3} constitutes an effective model. Its noise could be originated
as a coarse-grained effect of quantum fluctuations~\cite{dms} or inhomogeneity of the distribution of matter in the universe~\cite{br},
or perhaps as a combination of these and other factors.
While approaching this question may open the door to improve this model, herein we do not enter into these physical considerations and
focus on the mathematical properties of equation~\eqref{cos3}.

The remainder of this article is as follows. In section \ref{secbuge} we compare the consequences of making precise \eqref{cos3}
via either the Stratonovich or the It\^o prescription; in particular, the Stratonovich interpretation of noise (the one tacitly assumed
in \cite{john}) shows finite time blow-ups of the Hubble parameter with positive probability.
In contrast, the It\^o prescription shows global existence almost
surely and the expected asymptotic behavior. In section \ref{seccon} we draw our main conclusions.
Finally, in the Appendix, we give a concise summary on the theory of stochastic
dynamical systems to cover some of the results employed in this work.

\section{Blow-up versus global existence}\label{secbuge}
\subsection{Stratonovich equation}

Equation \eqref{cos3} can be written in the precise manner:
\begin{equation}
\label{cosstr}
dH=-H^2dt-\frac{3}{2}H^2 \circ dW_t,
\end{equation}
if we assume the stochastic integral of Stratonovich. This is tacitly assumed by the authors in \cite{john}, who employ
the usual Leibnitz-Newton calculus rules with this equation, known to be valid only in the case of Stratonovich.
Due to the regularity of its terms, this equation will possess a unique strong solution as long as it remains bounded \cite{oksendal}.
If we further follow reference \cite{john} by \emph{formally} performing the change of variables $x:=\frac{1}{H}$, then \eqref{cosstr} transforms
into
\begin{eqnarray}\nonumber
dx &=& dt+\frac{3}{2}\circ dW_t \\ \label{cos4}
&=& dt+\frac{3}{2}dW_t,
\end{eqnarray}
where we have just employed the Stratonovich (or Leibnitz-Newton) calculus rules on the first line and the fact that
an additive noise equation has not to be interpreted (at least, in other sense rather than the one provided by the Wiener integral).
For this stochastic differential equation (SDE) we have the following result:

\begin{lemma}
The SDE \eqref{cos4} subject to an initial condition $x(0)=1/H_0 >0$ possesses a unique solution which is both strong and global and moreover fulfils
$$
\lim_{t\rightarrow \infty}\frac{x(t)}{t}=1 \quad \text{a.s.}
$$
\end{lemma}

\begin{proof}
Existence and uniqueness of a strong and global solution is direct from the standard theory \cite{oksendal}.
To find the asymptotic behavior multiply the SDE by $2/3$ to get
$$
d \left(\frac{2x}{3} \right) = \frac{2}{3} dt + dW_t,
$$
which is of the form \eqref{sdera} with $f=2/3$ (see the Appendix).
Now we will apply to this equation the developments in the last part of the Appendix;
first, doing the change of variables \eqref{sdera0}, which in the present case reads
$$y=\frac{2}{3}x-z^{\ast},$$
leads to \eqref{sdera1}, which now takes the form
$$y(t)=\frac{2}{3}t+\lambda\int_0^t z^{\ast} d\tau .$$
It defines a random dynamical system \cite{arnold} and, by Proposition \ref{3.3.2}, it gives rise to a conjugated random dynamical
system for equation \eqref{cos4}. Subsequently applying Proposition \ref{pro331} yields
$$
\lim_{t\rightarrow \infty}\frac{y(t)}{t}= \frac{2}{3} \quad \text{a.s.}
$$
The statement follows from undoing the change of variables to recover $x$ and applying Proposition \ref{pro331} once more.
\end{proof}

From this result one would be tempted to conclude
\begin{equation*}
Ht\rightarrow1\ \text{a.s. when }t\rightarrow\infty,
\end{equation*}
which is in fact the behavior found for the deterministic model in the Introduction.
However, this conclusion would go through undoing the change of variables
$H=1/x$, which is possible only if $x(t)$ does not cross (or even touch) zero. Unfortunately, this is not the case as we show in the next result.

\begin{lemma}
The stochastic process $x(t)$ that solves SDE \eqref{cos4} subject to $x(0)=1/H_0 >0$ crosses the origin
in finite time with a probability
$$
\exp\left(-\frac{8}{9 H_0}\right).
$$
\end{lemma}

\begin{proof}
Note that we have the explicit formula $x(t) = t+\frac{3}{2}W_t+\frac{1}{H_0}$.
Therefore the statement can be seen as a particular case of the first passage time problem for Brownian motion with drift solved in \cite{karatzas}
by means of the Girsanov theorem. Consider the stopping time
$$
T_a:=\inf \left\{t\geq0:\frac{2}{3}t+W_t=a \right\},
$$
for any $a \neq 0$. It possesses the density \cite{karatzas}:
\begin{equation*}
\mathbb{P}(T_a\in dt)=\frac{|a|}{\sqrt{2\pi t^3}}\exp\left(-\frac{(a-\frac{2}{3}t)^2}{2t}\right)dt,\ t>0,
\end{equation*}
and therefore
\begin{equation*}
\mathbb{P}(T_a<\infty)=\exp\left(\frac{2}{3}a-\left|\frac{2}{3}a\right|\right).
\end{equation*}
As in our case $a=-\frac{2}{3 H_0}$ we conclude
\begin{equation*}
\mathbb{P} \left(T_{-\frac{2}{3 H_0}}<\infty \right)=\exp\left(-\frac{8}{9 H_0}\right).
\end{equation*}
\end{proof}

Since $\exp\left(-\frac{8}{9 H_0}\right) \in (0,1)$ for $H_0 \in (0,\infty)$, there is a positive probability that $H(t)$ blows up in finite time.
So, in summary, we obtain the following dual behavior for the solution to SDE \eqref{cosstr}:
\begin{eqnarray*}
H(t)t &\rightarrow& 1 \quad \text{as} \quad t \to \infty \qquad \text{with probability} \quad 1-\exp\left(-\frac{8}{9 H_0}\right), \\ \nonumber
H(t) &\rightarrow& \infty \quad \text{as} \quad t \to T_{-\frac{2}{3 H_0}} \qquad \text{with probability} \quad \exp\left(-\frac{8}{9 H_0}\right);
\end{eqnarray*}
i.e., the solution blows up in finite time with a positive probability $\exp\left(-\frac{8}{9 H_0}\right)$, so the solution ceases to exist in
finite time and therefore the asymptotic behavior is never achieved for these samples. Consequently, we may conclude that this model
presents difficulties in its application to cosmology (at least, in the absence of empirical evidence in favour of such divergence).

\begin{remark}
Despite the divergence, one could be tempted to use the formula for the Hubble parameter
$$
H(t)=(t + 3 W_t/2 + 1/H_0)^{-1}
$$
to define a scale factor
$$
C(t)=\exp \left[ \int_{0}^{t} H(s) ds \right].
$$
If this expression were well-defined, this would open the possibility to somehow reformulate the Stratonovich problem
to still get an almost surely meaningful scale factor. However, one can apply the mathematical framework proven in~\cite{peres}
to find that the singular set of such a Hubble parameter would have Hausdorff dimension $1/2$ with positive probability,
and therefore an uncountable number of divergences. In consequence we cannot make sense of this formal expression for the scale factor
by means of standard improper integration.
\end{remark}

\begin{remark}
Another way out of the problem is to impose a reflecting boundary condition at a sufficiently large value for the Hubble parameter. This would of course
stop the divergence, but one would need to justify the presence of this boundary on physical terms. In the absence of such a justification, it looks
simpler to us to just use the It\^o interpretation, which is free of such a need, as we will show in the next section. Perhaps more hopeful is to
use such a problem to build a sequence of processes that approximates the Stratonovich $H(t)$ as the upper boundary moves towards infinity. The expression
for the scale factor is well-defined for each of these reflected processes and thus one could build a sequence of scale factors out of the sequence
of reflected Hubble parameters. If this sequence converged in some suitable sense to a stochastic process one could, at least mathematically speaking,
define the scale factor as the limit. If this program could be carried out successfully, it would lead to a well-defined scale factor constructed
from an ill-defined Hubble parameter. Since we are not sure of a precise mathematical implementation and
an intuitive physical interpretation of such an approximation procedure, we leave this question open in the present work.
\end{remark}

\begin{remark}
Yet another way out of the problem caused by the finite time divergence is to assume the existence of a multiverse. Each
universe in this multiverse would be characterized by a different realization of the noise, and our universe would be one in which
the divergence does not happen (since indeed the Stratonovich model is divergence free with a positive probability). Because multiverse
physics is at this moment in a speculative stage, we will not discuss more this possibility.
\end{remark}

\subsection{It\^o equation}

Equation \eqref{cos3} can also be written in the different precise manner:
\begin{equation}
\label{cosito}
dH=-H^2dt-\frac{3}{2}H^2 dW_t,
\end{equation}
if we assume the stochastic integral of It\^o.
As for the Stratonovich model, the terms are so regular that
this equation will possess a unique strong solution as long as it remains bounded \cite{oksendal}.
If in this case we change variables again $x=1/H$ then, by the It\^o chain rule, we find
\begin{equation}
\label{cos8}
dx=\left(1+\frac{9}{4x}\right)dt+\frac{3}{2}dW_t.
\end{equation}
Note that this equation, again due to the regularity of its terms,
will possess a unique strong solution as long as it remains bounded and bounded away from zero \cite{oksendal}.
We will show that it is well behaved in the sense we are looking for. In order to analyze it we need to introduce the following family of stochastic
processes.

\begin{definition}\label{bess1}
For any integer $n\geq2$, let $\tilde{W}_t=(W_t^1,\cdots, W_t^n), t\geq0$, be a $n$-dimensional Wiener process. Then the process
$R_t= |W_t|, t\geq0$, is called a Bessel process of dimension $n$.
\end{definition}

\begin{remark}\label{bess2}
A Bessel process of dimension $n$ satisfies the SDE
$$
dR_t= \frac{n-1}{2R_t}dt+ dW_t, \quad R_0=0,
$$
where $W_t$ is the standard one-dimensional Wiener process, see \cite{karatzas}.
\end{remark}

\begin{theorem}\label{main}
The SDE \eqref{cos8} subject to an initial condition $x(0)=1/H_0 >0$ possesses a unique solution which is both strong and global and moreover fulfils
$$
\lim_{t\rightarrow \infty}\frac{x(t)}{t}=1 \quad \text{a.s.}
$$
Furthermore, $x(t)>0$ for all $t \geq 0$ a.s.
\end{theorem}

\begin{proof}
First of all change variables $z=\frac{2}{3}x$ to obtain
\begin{equation*}
dz=\frac{2}{3}dt+\frac{1}{z}dt+dW_t.
\end{equation*}
Note that the SDE
\begin{equation}
\label{cos9R}
dy=\frac{1}{y}dt+dW_t,
\end{equation}
subject to $y(0)=0$
possesses as unique solution the Bessel process $y(t)=R_t$ with dimension $n=3$, see Definition \ref{bess1} and Remark \ref{bess2}.
Then $y(t)>0$ for all $t >0$ almost surely, see \cite{karatzas}, and consequently
$z(t)>0$ for all $t >0$ almost surely provided that, for the same realization of $W_t$, $z(t) \geq y(t)$ for all $t >0$ almost surely.
Note that this should be true at least in some interval $[0,\mathfrak{t}]$,
for some positive stopping time $\mathfrak{t}=\mathfrak{t}(\omega)$, given the
continuity of the trajectories of both $z(t)$ and $y(t)$, and the fact that $z(0)=\frac{2}{3 H_0} > 0 = y(0)$.
Suppose that, on the contrary, there exists a set $\tilde{\Omega}\subset\Omega$ of positive measure such that for all $\omega\in \tilde{\Omega}$
there is a random time $t_{\omega}>\mathfrak{t}>0$
such that $z(t_{\omega},\omega)< y(t_{\omega},\omega)$ while $z(t,\omega)>0$ for all $t \in [0,t_{\omega}]$.
Moreover, by the continuity of the trajectories of both stochastic processes, and given their initial conditions,
there should exist a second random time $t_{\omega}'$, with $0< \mathfrak{t} \leq t_{\omega}'<t_{\omega}$,
such that the inequality $z(t,\omega)< y(t,\omega)$
holds for all $t \in (t_{\omega}',t_{\omega}]$ while $z(t_{\omega}',\omega) = y(t_{\omega}',\omega)$.
Therefore
\begin{flalign*}
z(t_{\omega},\omega)-y(t_{\omega},\omega)&=\frac{2(t_{\omega}-t_{\omega}')}{3}
+\int_{t_{\omega}'}^{t_{\omega}}\left(\frac{1}{z(\tau,\omega)}-\frac{1}{y(\tau,\omega)}\right)d\tau\\
&>\frac{2(t_{\omega}-t_{\omega}')}{3} >0,
\end{flalign*}
in contradiction with $z(t_{\omega},\omega)< y(t_{\omega},\omega)$.
This proves $z(t)$ is positive for all times almost surely and, moreover, that it is lower bounded by the Bessel process with dimension $n=3$.

To prove global boundedness and the asymptotic behavior we proceed again by means of comparison arguments. In particular, we claim that
\begin{equation}
\label{cos9RRR}
z_{-}(t)\leq z(t)\leq z_{+}(t),
\end{equation}
for all $t \geq 0$ almost surely, where
$$z_{-}(t):=\frac{2}{3}t+W_t+z(0),$$
and
$$z_{+}(t):=y(t)+\frac{2}{3}t+z(0),$$
where we, of course, always assume the same realization of $W_t$.
Indeed, in the first case we have
$$z(t,\omega)-z_{-}(t,\omega)=\int_0^{t}\frac{1}{z(\tau,\omega)}d\tau\geq0,$$
almost surely, due to the almost sure positivity of $z(t)$.
While in the second case we find
$$z_{+}(t,\omega)-z(t,\omega)=\int_0^{t}\left(\frac{1}{y(\tau,\omega)}-\frac{1}{z(\tau,\omega)}\right)d\tau\geq0,$$
almost surely since $z(t) \geq y(t)$ almost surely.
Now, since $\frac{W_t}{t}\rightarrow0$ as $t\rightarrow\infty$ almost surely by Proposition \ref{pro331}, then
$$\frac{y(t)}{t}=\sqrt{\left(\frac{W^1_t}{t}\right)^2+\left(\frac{W^2_t}{t}\right)^2+\left(\frac{W^3_t}{t}\right)^2}\rightarrow0\ \text{a.s.}$$
as $t\rightarrow\infty$ by Definition \ref{bess1}. Finally, from \eqref{cos9RRR} we conclude
$$
\lim_{t\rightarrow \infty}\frac{z(t)}{t}= \frac{2}{3} \quad \text{a.s.},
$$
and the statement follows from undoing the change of variables $x=\frac{3}{2}z$.
\end{proof}

From Theorem \ref{main} we conclude that the change of variables $x(t)=1/H(t)$ is well defined for all $t \geq 0$ almost surely, so we can invert it
to conclude
\begin{equation*}
H(t)t \rightarrow 1 \quad \text{as} \quad t \to \infty \qquad \text{with probability} \quad 1.
\end{equation*}

\section{Conclusions}\label{seccon}

In this work we have studied a stochastic version of the Friedmann acceleration equation. This SDE was introduced in \cite{john} as a possible
theoretical explanation of the uncertainty observed in the experimental quantification of the Hubble parameter. In the original model,
the Stratonovich interpretation of noise was tacitly chosen, as the authors analyzed this equation using the standard chain rule
from ordinary calculus. Herein we have shown that this Stratonovich equation presents finite time blow-ups with a positive probability;
in particular this probability is
$$
\exp\left(-\frac{8}{9 H_0}\right).
$$
On the other hand, we have shown that if the interpretation is changed to that of It\^o, then the solution exists globally in time almost surely
and moreover we have the asymptotic behavior
$$
H(t)t \rightarrow 1 \quad \text{as} \quad t \to \infty \qquad \text{with probability} \quad 1,
$$
or, in other words, we have for almost every sample the universal decay rate $H(t) \approx 1/t$, with universal amplitude, for long times.
This coincides with both the deterministic behavior (which of course takes place surely) and the long time dynamics of the Stratonovich
SDE in those cases in which the solution does not blow up in finite time. Yet another advantage of the It\^o interpretation is that
it admits the explicit bounds
$$
\frac{1}{\frac{3}{2}y(t)+t+\frac{1}{H_0}} \leq H(t) \leq \frac{1}{\max\{\frac{3}{2}W_t+t+\frac{1}{H_0},\frac{3}{2}y(t)\}}
$$
for all times $t \geq 0$ almost surely, where $y(t)$ is the Bessel process with dimension $n=3$ that solves equation \eqref{cos9R};
this is a direct consequence of the proof of Theorem \ref{main}.

Although it seems that the Stratonovich interpretation is more accepted for applications in physics, we have herein highlighted an example
in the field of cosmology in which the It\^o interpretation looks more advantageous.
This is not the only case, as other problems in statistical mechanics are somehow
similar in this respect \cite{ce,escudero,escudero2}. In these previous cases, however, the uniqueness of solution was lost in the Stratonovich
equation, while preserved in the It\^o one. In the present case, it is the (global in time) existence of solution what is lost for Stratonovich
and preserved for It\^o. Another difference between these previous statistical mechanical examples and the current cosmological one is that
the deficiency of the Stratonovich equation was of a more fundamental character there, as there is a Stratonovich equation equivalent
to equation \eqref{cosito} due to the smoothness of its diffusion; this did not happen in the other cases because the equations under study there
did not present regular enough diffusion terms. Anyway, in the problem we have analyzed herein, it seems much simpler to employ the It\^o
interpretation rather than try to circumvent it via an equivalent Stratonovich formulation. In summary, we have illustrated one example in
cosmology for which It\^o versus Stratonovich means global existence versus finite time blow-up.
Of course, this does not mean that physical arguments are not important in the derivation of stochastic differential models. They are fundamental and,
in fact, we have used them in the derivation of the models considered herein. However, the mathematical consistency of a physically derived
model must be also checked. In those cases in which basic properties such as existence and uniqueness of solution fail, either there is a physical
reason for this failure, or the model needs to be rectified.

\section*{Acknowledgments}

This work has been partially supported by the Government of Spain (Ministerio de Ciencia, Innovaci\'on y Universidades)
through Project PGC2018-097704-B-I00.

\section*{Appendix: Stochastic dynamical systems}

In this section we provide a summary of the mathematical framework we use, that of random and stochastic dynamical systems.
To this end we follow \cite{arnold}, \cite{caraballo}, and \cite{crauel}. We start introducing the definition of
random dynamical system. We use the notation $\mathbb{R}_+:=[0,\infty)$.

\begin{definition}\label{dds}
Let $(\Omega,\mathcal{F},\mathbb{P})$ be a probability space and $\theta=\{\theta_t\}_{t\in\mathbb{R}}$ a map
$\theta:\mathbb{R}\times\Omega\mapsto\Omega$, satisfying
\begin{enumerate}
\item $\theta_0(\omega)=\omega,\ \forall \omega\in\Omega$.
\item $\theta_{t+\tau}(\omega)=\theta_t\circ\theta_{\tau}(\omega),\ \forall t,\tau\in\mathbb{R}_+$.
\item $(t,\omega)\mapsto \theta_{t}(\omega)$ is measurable in $(\mathcal{B}(\mathbb{R})\times\mathcal{F},\mathcal{F})$, and $\theta_t\mathbb{P}=\mathbb{P}$ for all $t\in\mathbb{R}$ (i.e. $\mathbb{P}(\theta_t(\omega))=\mathbb{P}(\omega)$).
\end{enumerate}
Then $(\Omega,\mathcal{F},\mathbb{P};\theta)$ is called a driving dynamical system and $\theta$ is called a driver.
\end{definition}
\begin{definition}
Let $(\Omega,\mathcal{F},\mathbb{P})$ be a probability space. A random dynamical system (RDS)
over $\mathbb{R}^d$ ($d=1,2,3,\cdots$) consists on a pair $(\theta,\varphi)$ with $\theta$ a driver
and a cocycle mapping $\varphi:\mathbb{R}_{+}\times\Omega\times\mathbb{R}^d\mapsto\mathbb{R}^d$ satisfying:
\begin{enumerate}
\item $\varphi(0,\omega,x)=x, \forall \omega\in\Omega, x\in\mathbb{R}^d$.
\item $\varphi(t+s,\omega,x)=\varphi(t,\theta_s(\omega),\varphi(s,\omega,x)),\ \forall t,s\in\mathbb{R}_{+}, \omega\in\Omega, x\in\mathbb{R}^d$.
\item $(t,\omega,x)\mapsto\varphi(t,\omega,x)$ is measurable.
\item $x\mapsto\varphi(t,\omega,x)$ is continuous $\forall t\in \mathbb{R}_{+}$ a.s.
\end{enumerate}
\end{definition}
Henceforth we fix $d=1$ and take the sample space $\Omega=C_0(\mathbb{R},\mathbb{R})$, the space of continuous functions
that depend on a real variable, take values on the reals, and vanish at the origin; let $\mathcal{F}$ be its Borel sigma-field.
Moreover we take $\mathbb{P}$ to be the Wiener measure, $\theta$ the Wiener shift
$$
\theta_t \omega(\cdot) = \omega(t + \cdot) - \omega(t),
$$
and $\omega(t)$ two-sided Brownian motion, with the usual definition (two independent Wiener processes, respectively
with positive and negative time, glued at the origin).
In this setting consider the one dimensional linear stochastic differential equation (SDE)
\begin{equation}
\label{ou1}
dz(t)=-\lambda z(t)dt+dW_t,
\end{equation}
for some $\lambda>0$ and $W_t(\omega)=\omega(t)$. Note that this is an Orstein-Uhlenbeck SDE.
Then, the following result holds \cite{caraballo}.
\begin{proposition}
\label{pro331}
Let $\lambda>0$. Then there exists a $\{\theta_t\}_{t\in\mathbb{R}}$-invariant
(in the sense of (3) in Definition~\ref{dds}) subset $\overline{\Omega}\in\mathcal{F}$ of $\Omega=C_0(\mathbb{R},\mathbb{R})$
with unit measure such that for all $\omega \in \overline{\Omega}$:
\begin{enumerate}
\item $\lim_{t\rightarrow\pm\infty}\frac{\omega(t)}{t}=0$.
\item The random variable given by $z^{\ast}(\omega):=-\lambda\int_{-\infty}^{0}e^{\lambda \tau}\omega(\tau)d\tau$ is well defined, and the map
$$(t,\omega)\mapsto z^{\ast}(\theta_t\omega)=-\lambda\int_{-\infty}^{0}e^{\lambda\tau}\theta_{t}(\omega(\tau))d\tau=-\lambda\int_{-\infty}^{0}e^{\lambda\tau}\omega(t+\tau)d\tau+\omega(t),$$
is a stationary solution of \eqref{ou1}, with continuous trajectories satisfying:
\begin{enumerate}
\item $\lim_{t\rightarrow\pm\infty}\frac{|z^{\ast}(\theta_t\omega)|}{|t|}=0$,
\item $\lim_{t\rightarrow\pm\infty}\frac{1}{t}\int_0^tz^{\ast}(\theta_{\tau}\omega)d\tau=0$,
\item $\lim_{t\rightarrow\pm\infty}\frac{1}{t}\int_0^t|z^{\ast}(\theta_{\tau}\omega)|d\tau<\infty$.
\end{enumerate}
\end{enumerate}
\end{proposition}
\begin{proof}[Proof]
We just give a sketch of the proof and refer the reader to the references for more details.
\begin{enumerate}
\item follows from the law of iterated logarithms (see for instance \cite{karatzas}):
$$\mathbb{P}\left\{\limsup_{t\rightarrow\infty}\frac{\omega(t)}{\sqrt{2t\log\log t}}=1\right\}=\mathbb{P}\left\{\limsup_{t\rightarrow\infty}\frac{\omega(t)}{\sqrt{2t\log\log 1/t}}=1\right\}=1.$$
\item is a consequence of the fact that the process $(\omega,t)\mapsto\int_{-\infty}^{t}e^{-\lambda(t-\tau)}dW_{\tau}(\omega)$
is a solution to \eqref{ou1} with $z(0)=0$ and of integration by parts.
The stationarity of this solution follows from the invariance of the Wiener measure $dW_t$
with respect to the flow $\{\theta_t\}_t$. The other statements follow from the ergodic theorem and the Burkholder inequality
(see for instance \cite{karatzas}).
\end{enumerate}
\end{proof}

At this point, we will employ the theory of conjugated RDSs to SDEs with additive noise. To that end we will need the following result \cite{caraballo}.

\begin{proposition}
\label{3.3.2}
Let $(\theta,\varphi)$ be a RDS on $\mathbb{R}$. Suppose that the mapping $T:\Omega\times\mathbb{R}\mapsto\mathbb{R}$ possesses the following properties:
\begin{enumerate}
\item The mapping $T(\omega,\cdot)$ is a homeomorphism over $\mathbb{R}$ for every $\omega\in \Omega$.
\item The mappings $T(\cdot,x)$ and $T^{-1}(\cdot,x)$ are measurable for every $x\in\mathbb{R}$.
\end{enumerate}
Then, the mapping $(t,\omega,x)\mapsto\phi(t,\omega,x):=T^{-1}(\theta_t(\omega),\varphi(t,\omega,T(\omega,x)))$
defines a conjugated random dynamical system.
\end{proposition}
\begin{proof}
First, note that
$$T^{-1}(\theta_0(\omega),\varphi(0,\omega,T(\omega,x)))=T^{-1}(\omega,T(\omega,x))=x.$$
Then, let us see that the property of cocycles is fulfilled:
\begin{flalign*}
\phi(t+\tau,\omega,x)&=T^{-1}(\theta_{t+\tau}(\omega),\varphi(t+\tau,\omega,T(\omega,x)))\\
&=T^{-1}(\theta_{t}(\theta_{\tau}(\omega)),\varphi(t,\theta_{\tau}(\omega),\varphi(\tau,\omega,T(\omega,x)))\\
&=\phi(t,\theta_{\tau}(\omega),\phi(\tau,\omega,x)).
\end{flalign*}
\end{proof}

Now, following \cite{caraballo}, we apply this theory to a SDE with additive noise of the form
\begin{equation}
\label{sdera}
dx(t,\omega)=f(x)dt+dW_t,\ x(0,\omega)=x_0\in\mathbb{R},
\end{equation}
where we suppose that $f$ is continuously differentiable with bounded derivative,
so that a unique strong solution to \eqref{sdera} exists almost surely (see \cite{oksendal} for milder conditions for existence and uniqueness).
Subsequently perform the change of variables
\begin{equation}
\label{sdera0}
y(t,\omega)=x(t,\omega)-z^{\ast}(\theta_t(\omega)),
\end{equation}
where $z^{\ast}$ is defined as in Proposition \ref{pro331}. Then it holds that
\begin{equation*}
y(0,\omega)=x(0,\omega)-z^{\ast}(\omega)=x_0-z^{\ast}(\omega).
\end{equation*}
So we obtain
\begin{flalign*}
dy(t,\omega)&=dx(t,\omega)-dz^{\ast}\\
&=f(x(t,\omega))dt+dW_t(\omega)-(-\lambda z^{\ast}(\theta_t(\omega))dt+dW_t(\omega))\\
&=(f(x(t,\omega))+\lambda z^{\ast}(\theta_t(\omega)))dt\\
&=(f(y(t,\omega)+z^{\ast}(\theta_t(\omega)))+\lambda z^{\ast}(\theta_t(\omega)))dt,
\end{flalign*}
and thus
\begin{equation}
\label{sdera1}
\frac{dy}{dt}(t,\omega)=f(y(t,\omega)+z^{\ast}(\theta_t(\omega)))+\lambda z^{\ast}(\theta_t(\omega)).
\end{equation}
Equation \eqref{sdera1} is a random differential equation, and $f$ is regular enough so that it generates a RDS $\varphi$ \cite{arnold};
then from Proposition \ref{3.3.2} we obtain a conjugated RDS for equation \eqref{sdera}.
Indeed, in this case, the map $T:\Omega\times\mathbb{R}\mapsto\mathbb{R}$ is given by
$$T(\omega,x):=x-z^{\ast}(\omega),$$
which satisfies the hypotheses of Proposition \ref{3.3.2}. Therefore,
\begin{flalign*}
\phi(t,\omega,x_0)&=T^{-1}(\theta_t(\omega),\varphi(t,\omega,T(\omega,x_0)))\\
&=T^{-1}(\theta_t(\omega),\varphi(t,\omega,x-z^{\ast}(\omega)))\\
&=y(t,\omega,x_0-z^{\ast}(\omega))+z^{\ast}(\omega)\\
&=x(t,\omega,x_0),
\end{flalign*}
is our conjugated RDS for the SDE \eqref{sdera}.

\vskip10mm
\noindent
{\footnotesize
Carlos Escudero\par\noindent
Departamento de Matem\'aticas Fundamentales\par\noindent
Universidad Nacional de Educaci\'on a Distancia\par\noindent
{\tt cescudero@mat.uned.es}\par\vskip1mm\noindent
}
\vskip1mm
\noindent
{\footnotesize
Carlos Manada\par\noindent
Facultad de Ciencias\par\noindent
Universidad Nacional de Educaci\'on a Distancia\par\noindent
{\tt cmanada1@alumno.uned.es}\par\vskip1mm\noindent
}

\end{document}